\documentclass[conference,letterpaper]{IEEEtran}
\IEEEoverridecommandlockouts
\usepackage{cite}
\usepackage{amsmath,amssymb,amsfonts}
\usepackage{amsthm}
\newtheorem{remark}{Remark}

\newtheorem{theorem}{Theorem}
\usepackage{algorithmic}
\usepackage{graphicx}
\usepackage{textcomp}
\usepackage{xcolor}
\usepackage{mathrsfs}
\usepackage{balance}
\def\BibTeX{{\rm B\kern-.05em{\sc i\kern-.025em b}\kern-.08em
    T\kern-.1667em\lower.7ex\hbox{E}\kern-.125emX}}
\usepackage{amssymb}

\usepackage{tikz}

\newcommand{\circbar}{
  \mathrel{
    \tikz[baseline=-0.5ex]{
      \draw (0,0) circle (0.4ex);
      \draw (-1ex,0) -- (1ex,0);
    }
  }
}

\allowdisplaybreaks

\begin{document}

\title{Perfectly Private Over-the-Air Computation
%\\ for Federated Learning
\thanks{This work was funded in part by the Swedish Foundation for Strategic Research (FUS21-0026) and in part by the Swedish Innovation Agency (Vinnova) through the SweWIN Center (2023-00572).}
}
\author{\IEEEauthorblockN{Shudi Weng, Ming Xiao, Mikael Skoglund}
\IEEEauthorblockA{\textit{Department of Information Science and Engineering} \\
\textit{KTH Royal Institute of Technology}\\
Stockholm, Sweden \\
\{shudiw, mingx, skoglund\}@kth.se
}
}

\maketitle
\begin{abstract}
This paper studies a key research question: how to achieve perfect privacy in over-the-air computation (AirComp)? The problem is particularly intriguing due to a dilemma. Real-field operations can ensure invertibility but generally introduce statistical dependence, resulting in inevitable privacy leakage. In contrast, modulo operations can decorrelate the output from the original message, but suffer from the ill-posed invertibility when applied 
over non-prime groups (e.g., the real field).
This raises a subtle yet fundamental question: Does perfect privacy intrinsically conflict with AirComp?
We show that the answer is no.
By carefully leveraging the interplay between real-field and modulo operations, perfect privacy and accurate computation can, in fact, be achieved simultaneously, enabling perfectly private aggregation.
\end{abstract}


%a fundamental limitation: it induces a many-to-one mapping from the real field onto a compact domain (the torus), resulting in a loss of invertibility. 

%%%%%%%%%%------------------------------------
%\vspace{-0.5em}
\section{Introduction}
Over-the-air computation (AirComp) refers to the computation of mathematical functions by exploiting the waveform superposition property of multiple-access channels (MACs) \cite{10092857}. In contrast to the conventional communication-computation separation paradigm, 
where signals are transmitted orthogonally and computation is performed after decoding, AirComp enables simultaneous transmissions and direct function computation through in-air signal superposition. By leveraging interference rather than avoiding it, AirComp significantly improves communication efficiency and reduces bandwidth requirements. 
This paradigm has emerged as a key enabler for efficient large-scale distributed learning systems, where frequent aggregation of model updates is required \cite{zheng2025secure}.

Despite its advantages, AirComp raises critical privacy concerns, as the transmitted signal contains information about user data. 
%Ensuring privacy in AirComp is therefore essential for deploying trustworthy wireless learning systems. %\cite{huang2021evaluating,xue2025dictpflefficientprivatefederated}.
Existing approaches to private AirComp primarily rely on jamming, cryptographic techniques, and noise injection. %\cite{10264827,nordlund2025secureovertheaircomputationmultiple,feng2025over,weng2025codingenforcedrobustsecureaggregation,9929413,10161677}. 
Jamming-based approaches \cite{frey2021towards,yan2025secure,10681103} aim to obscure transmitted signals by introducing artificial interference. 
However, their effectiveness depends heavily on channel conditions and requires additional transmission resources, reducing overall system efficiency.
Cryptographic methods, such as homomorphic encryption (HE) \cite{wang2025efficient}, enable users to encrypt their messages prior to transmission and perform aggregation directly on ciphertexts, providing strong privacy guarantees. Nevertheless, the associated encryption and decryption procedures can incur substantial computational overhead on clients.
In contrast, noise injection methods apply additive masks to original messages before transmission, requiring only simple arithmetic operations and thus imposing minimal computational demands while providing a certain level of privacy protection. 
To this end, early approaches consider independently generated noise at each client \cite{10264827}. However, such noise is fully preserved and accumulated at the receiver, resulting in significant degradation in aggregation accuracy. 
To mitigate this issue, subsequent works introduce correlated noise designs across clients \cite{nordlund2025secureovertheaircomputationmultiple}, where the injected noise is structured to partially cancel out during in-air aggregation, thereby reducing the residual privacy noise.
However, both approaches remain constrained by the inherent privacy-utility trade-off. Even with privacy-accuracy optimization techniques \cite{feng2025over}, accuracy loss can be considerable to gain strong privacy. 
Building upon this idea, more recent approaches exploit zero-sum noise constructions \cite{weng2025codingenforcedrobustsecureaggregation,9929413}, ensuring that the aggregate noise vanishes exactly, thereby removing the entire privacy noise and preserving the in-air aggregation accuracy \cite{9929413,10161677}. 
%This significantly alleviates the privacy–utility tradeoff and makes zero-sum noise a promising direction for private AirComp. 
Beyond these, zero-sum noise has been demonstrated to be a promising approach for tackling various real-world challenges, e.g., user collusion, dropouts, and heterogeneous security requirements. \cite{zhao2022information,zhang2026fundamentallimitshierarchicalsecure,li2025weakly}.

Yet, existing noise injection-based private AirComp protocols \cite{10264827,nordlund2025secureovertheaircomputationmultiple,feng2025over,9929413,10161677} all suffer from nonzero privacy leakage, as real-field operations (e.g., addition and multiplication) introduce statistical dependence between inputs and outputs. 
While perfect privacy can, in principle, be asymptotically approached by increasing the power of privacy noise, this strategy is impractical in real-world AirComp systems. In particular, the large-magnitude transmit signals induced by high privacy noise power are susceptible to nonlinear distortion and clipping in power amplifiers, which could severely degrade aggregation accuracy.

These practical limitations raise a fundamental question: Is it possible to achieve perfect privacy in AirComp under realistic conditions (e.g., in the presence of channel fading, noise, and power constraints)?
A potential alternative is to employ modulo operations, which can decorrelate transmitted signals from the underlying data by introducing uniform noise. 
Despite this distinct advantage, the approach has an intrinsic limitation. 
In wireless transmission, the channel operates in the analog domain regardless of the transmitted signal representation, whereas modulo operations map real-valued signals onto a compact domain (the torus), resulting in a many-to-one and inherently non-invertible transformation.
As a result, the unique recovery of the original real-valued signals from their modulo representations is only possible under a restrictive condition, e.g., the unimodularity of the associated linear transformation.
%This reveals a fundamental conflict: real-field operations enable accurate computation but leak information; in contrast, modulo operations ensure privacy but render the inverse problem from the torus to the real field ill-posed. 
In AirComp, however, the channel coefficients are random, and rarely satisfy the unique recovery condition, potentially causing ambiguity in signal reconstruction and undermining the computation correctness. \footnote{To avoid any confusion, the signal recovery ambiguity discussed here is caused solely by the non-prime group size and is unrelated to the input range.}

%while modulo operations enable privacy by removing statistical dependence, they generally render the inverse problem from the torus to the real field ill-posed.
%the resulting inverse problem from the torus to the real field is generally ill-posed. Due to random channel realizations, this condition is rarely satisfied in practical AirComp systems, potentially resulting in ambiguity in signal recovery and undermining the correctness of AirComp . 

%The unique recovery of real-field signal is only possible under restrictive conditions, i.e., unimodularity of the linear transformation matrix of the torus coefficients. 

%%%%%%%%%%------------------------------------
This work aims to advance the frontier of private AirComp by introducing a cross-domain design that integrates real-field operations with modulo masking, thereby overcoming the utility-privacy trade-off under a finite power budget. To the best of our knowledge, this is the first work to achieve perfect physical-layer privacy for AirComp under realistic channel conditions without relying on cryptographic techniques\footnote{Despite \cite{brune2023private} achieves perfect privacy for AirComp through a lattice-based approach,
it
relies on a simplified channel model that neglects key practical challenges (e.g., fading), and thus falls outside the scope of this work.}.    
The proposed method supports both discrete and analog computation.
%by developing a perfectly private analog aggregation protocol without sacrificing accuracy under a finite power budget. 
Our main contributions are summarized below.
\begin{itemize}
    %\item We propose P$^2$-AirComp, a perfectly private AirComp protocol that enables accurate computation. 
    %P$^2$-AirComp combines and alternates the cross-domain operations to match the operational properties of real field and torus. 
        %and to exploit their complementary advantages.
        %by leveraging the interplay between real-field and modulo operations. 
        %In particular, t
        %This is achieved by carefully combining and alternating the cross-domain operations to match their operational properties and to exploit their complementary advantages. 
    %The real-field operations preserve a unique mapping during operations, enabling accurate aggregation and handling the channel effects, while the modulo operation ensures perfect user privacy through decorrelation. 
    %As a result, both their complementary advantages are preserved.
    \item We propose a private AirComp protocol, P$^2$-AirComp, that achieves both perfect privacy and accurate analog computation. 
    P$^2$-AirComp combines and alternates cross-domain operations to align the properties of the real-field and the modulo operations, thereby preserving their complementary advantages simultaneously. Theoretical analysis and simulation results show that the proposed method can achieve simultaneous accurate computation and perfect privacy.
    %both accurate computation and perfect privacy can be achieved simultaneously within a finite power budget.
    \item We show that the proposed P$^2$-AirComp achieves perfect information-theoretic security both in the presence of correlated secret keys at the clients and in the absence of any side information at the server.
    \item We analyze the mean square error (MSE) distortion of P$^2$-AirComp in terms of the effective noise, accounting for the impact of the modulo operation. We further derive a closed-form expression for the pointwise MSE, along with explicit upper and lower bounds.
    %Our analysis differs from existing ones, which are purely rooted in real-field operations.
    %\textcolor{red}{no-wrapping event with high probability?}
\end{itemize}

%%%%%%%%%%------------------------------------
\section{System Model}
%%%%%%%%%%------------------------------------
\subsection{AirComp Protocol} \label{sec: AirComp model}
%Our AirComp model refers to \cite{9929413}. 
Each client \(k\) holds a message vector \(\mathbf{e}_k \in \mathbb{R}^D\) to be sent to the server. Assume that the entries in $\mathbf{e}_k$ are identically and independently distributed (i.i.d.) under a total power constraint $\mathbb{E}\!\left[\lVert \mathbf{e}_k \rVert^2\right]\leq P_E$.
Let \(h_k \in \mathbb{R}\) denote the uplink channel coefficient from client \(k\) to the server at round \(t\), and let \(\mathbf{x}_k \in \mathbb{R}^D\) denote the transmitted signal. 
Assume perfect channel state information at the transmitter (CSIT), using channel inversion precoding, client \(k\) transmits
\begin{align}
    \mathbf{x}_k = \frac{\sqrt{P}}{h_k}\mathbf{e}_k,
    \qquad h_k \neq 0,
    \label{eq: tx_signal}¢
\end{align}
where \(P>0\) is a common power scaling factor and $\alpha_{k}\triangleq\frac{\sqrt{P}}{h_k}$ is the transmit scaling factor. 

Following this, the received signal at the server is
\begin{align}
    \mathbf{y}
    &= \sum_{k=1}^K h_k \mathbf{x}_k + \mathbf{z} \label{eq: y_signal1}\\
    &= \sqrt{P}\sum_{k=1}^K \mathbf{e}_k + \mathbf{z},
    \label{eq: y_signal2}
\end{align}
where the noise vector \(\mathbf{z} \in \mathbb{R}^D\) is independent of the transmit signal and follows $ \mathbf{z} \sim \mathcal{N}(\mathbf{0}, N_0 \mathbf{I}_D)$.
Moreover, the transmit power of client \(k\) is assumed to satisfy a total power constraint, $\mathbb{E}\!\left[\|\mathbf{x}_k\|^2\right]
    = \frac{P}{|h_k|^2}\mathbb{E}\!\left[\|\mathbf{e}_k\|^2\right]
    \le P_X$.
%%%%%%%%%-----------------------------------------------------------------------------------
\subsection{Threat Model}
Our threat model follows a standard assumption adopted in multi-party computation (MPC). The server and clients are assumed to be honest but curious: they faithfully execute the protocol but may attempt to infer other private messages from publicly observable information and their local knowledge. 
Specifically, a client may try to learn the private message by exploiting the broadcasts from other clients and its own local secret key. 
The server may attempt to infer any private message from its observations.

%%%%----------------------------------------------------------------------

%%%%%%%%%%------------------------------------
\section{The Proposed Method: P$^2$-AirComp}\label{sec: PP_OTA_FL} 
This section presents the proposed protocol, perfectly private over-the-air computation (P$^2$-AirComp). For clarity, we omit standard procedures from the system model and focus on the key novel aspects of our approach. 
%%%%%%%%%%------------------------------------
\subsection{Secret Key Generation}
Each client $k$ holds a secret key $\mathbf{S}_k \sim \mathrm{Unif}\left( (-\frac{1}{2},\frac{1}{2}]^D \right)$ to mask its private messege. %Specifically, each element of $\mathbf{S}_k$ is drawn from a uniform distribution over $[-\frac{1}{2},\frac{1}{2})$. 
To ensure that the maskings cancel out during aggregation, the secret keys are jointly constructed to satisfy the modulo-zero-sum condition,
\begin{align}
\left[\sum_{k=1}^K \mathbf{S}_k \right]\hspace{-2mm} \mod 1= \mathbf{0},
\label{eq:modulo_lero_sum_con}
\end{align}
where $\hspace{-2mm} \mod 1$ denotes the zero-centered modulo operation defined by $x \mod 1 \triangleq x - \left\lfloor x + \tfrac{1}{2} \right\rfloor$ and is applied element-wisely.
%denotes the element-wise modulo operation .
A general construction is to generate $J$ independent random vectors 
$\mathbf{N}_j \sim \mathrm{Unif}([-\frac{1}{2},\frac{1}{2})^D)$ for $j=1,\dots,J$, and construct $\{\mathbf{S}_k\}_{k=1}^K$ by a generator matrix $\mathbf{G}_S \in \mathbb{Z}^{K\times J}$ as
\begin{align}
    \begin{bmatrix}
        \mathbf{S}_1\\
        \vdots\\
        \mathbf{S}_K
    \end{bmatrix}
    =\mathbf{G}_S     
    \begin{bmatrix}
        \mathbf{N}_1\\
        \vdots\\
        \mathbf{N}_J
    \end{bmatrix} \hspace{-2mm} \mod 1,
\end{align}
where $\mathbf{G}_S$ satisfy the following two conditions:
\begin{enumerate}
    \item Any $d$ rows in $\mathbf{G}_S$ forms a full-rank matrix.
    \item All $K$ rows in $\mathbf{G}_S$ sum to $\mathbf{0}$, i.e, $\mathbf{1}_K^\top \mathbf{G}_S = \mathbf{0}^\top$, where $\mathbf{1}_K$ denote $K\times 1$ all-one vector.
\end{enumerate}
In this construction, each $\mathbf{S}_k\sim \mathrm{Unif}([-\frac{1}{2},\frac{1}{2})^D)$ and satisfies \eqref{eq:modulo_lero_sum_con}. Moreover, any $d$ secret keys are mutually independent. 

In this paper, we set $J=K-1$, ensuring that no subset of fewer than $K$ secret keys can cancel out. The only linear combination of the secret keys that results in full cancellation is given by \eqref{eq:modulo_lero_sum_con}. This property guarantees that no information about the private message other than the sum can be revealed, which will be discussed later.

%%%%%%%%%%------------------------------------
\subsection{P$^2$-AirComp}
Assume each client holds a private message $\mathbf{W}_k\in \mathbb{R}^D$,
and $\mathbf{W}\triangleq \sum_{k=1}^K \mathbf{W}_k\in [-a,a]^D$, where $0<a<\frac{1}{2}$. Each client $k$ first generates a signal based on its private message $\mathbf{W}_k$ and its equipped secret key $\mathbf{S}_k$,  
\begin{align}
\mathbf{e}_k=\left[\mathbf{W}_k+
\mathbf{S}_k\right] \hspace{-2mm} \mod 1.
\label{eq: tx_message}
\end{align}
Then, each client generates a transmit signal $\mathbf{x}_k$ by rescaling $\mathbf{e}_k$ as \eqref{eq: tx_signal}, and transmits $\mathbf{x}_k$ through the AirComp protocol in Section \ref{sec: AirComp model}, the superposed signal received at the server is given in \eqref{eq: y_signal1} and \eqref{eq: y_signal2}. In P$^2$-AirComp, the server receives
\begin{align}
    \mathbf{y}
    %=\sqrt{P}\sum_{k=1}^K \mathbf{e}_k + \mathbf{z}
    &= \sqrt{P}\sum_{k=1}^K \left( \left[\mathbf{W}_k+
\mathbf{S}_k\right] \hspace{-2mm} \mod 1\right)  + \mathbf{z}.
\end{align}
Then, the server estimates the sum by 
\begin{align}
    \hat{\mathbf{W}}
    &=\left[\frac{\mathbf{y}}{\sqrt{P}}\right] \hspace{-2mm} \mod 1\\
    %&=\left[\sum_{k=1}^K \left[\frac{1}{K}\mathbf{W}_k +\mathbf{S}_k\right] \hspace{-2mm} \mod 1 + \frac{\mathbf{z}}{\sqrt{P}} \right] \hspace{-2mm} \mod 1\\
    &=\left[\sum_{k=1}^K \left(\mathbf{W}_k +
\mathbf{S}_k\right) + \frac{\mathbf{z}}{\sqrt{P}} \right] \hspace{-2mm} \mod 1 \label{eq: server key step}\\
    &= \left[\sum_{k=1}^K \mathbf{W}_k+ \frac{\mathbf{z}}{\sqrt{P}}\right] \hspace{-2mm} \mod 1, \label{eq: noisy sum}
\end{align}
where \eqref{eq: server key step} applies the modulo additivity property, i.e., 
$\left[\sum_{k=1}^K [x_k] \mod 1 \right] \mod 1=\left[\sum_{k=1}^K x_k \right] \mod 1.$

%For privacy analysis, we guide the reader to Section \ref{sec: Privacy}.  

%In this design, the estimated sum $\hat{\mathbf{W}}$ may deviate from the true value $\mathbf{W}$ for two reasons. First, receiver-side noise introduces distortion. Second, if this distortion becomes sufficiently large, the modulo operation may cause additional distortion due to wrapping. The former effect is also present in conventional OTA, whereas the latter arises specifically from the incorporation of the modulo operation in OTA.
%For distortion analysis, we guide the reader to Section \ref{sec: distortion-snr}. 

In this regime, the power constraints are  
\begin{align}
&\mathbb{E}\!\left[\lVert \mathbf{e}_k \rVert^2\right]=\frac{1}{12}\leq P_E,\quad\quad\hspace{2mm} k\in [K],\\
&\mathbb{E}\!\left[\|\mathbf{x}_k\|^2\right]=\frac{P}{h_k^2}P_E\leq P_X,\quad  k\in [K].   
\end{align}
Following this, the common power scaling factor is limited by $P\leq\min_{k} \left\{\frac{P_X}{P_E} h_k^2\right\} $. 
Through this constraint, the channel coefficient implicitly affects the estimation result in \eqref{eq: noisy sum}, even though the channel coefficient does not appear explicitly. 

By carefully aligning the operational properties of real-field and modulo operations, the proposed scheme achieves simultaneous accurate summation and perfect privacy. 
In particular, modulo operation ensures that each transmitted signal is statistically independent of the underlying private message, thereby achieving perfect privacy. Real-field aggregation helps to avoid potential ambiguity, enabling correct computation.
Notably, this is achieved under a finite transmit power budget. Such advantages are not attainable by existing methods.

%%%%%%%%%%------------------------------------
%\vspace{-0.5em}
\section{Privacy Analysis} \label{sec: Privacy}
\subsection{Security at the Server}
\begin{theorem}[Perfect Privacy at the Server]\label{theo: server security}
In P$^2$-AirComp, the server obtains no information about the private messages beyond the aggregation, i.e., $I( \left\{\mathbf{W}_k\right\}_{k\in[K]}; \mathbf{y} | \mathbf{W} )=0$.   
\end{theorem}
\begin{proof} The information leakage can be quantified by
    \begin{align}
    &I\left( \left\{\mathbf{W}_k\right\}_{k\in[K]}; \mathbf{y} \middle| \mathbf{W} \right) \notag\\
    \leq&I\left( \left\{\mathbf{W}_k\right\}_{k\in[K]}; \left\{\mathbf{x}_k\right\}_{k\in[K]}, \mathbf{z} \middle| \mathbf{W} \right) \label{eq:server_sec_step0}\\
    =&I\left( \left\{\mathbf{W}_k\right\}_{k\in[K]}; \left\{\mathbf{x}_k\right\}_{k\in[K]} \middle| \mathbf{W} \right)  \label{eq:server_sec_step1}\\
    =&I\left( \left\{\mathbf{W}_k\right\}_{k\in[K]}; \left\{\mathbf{e}_k\right\}_{k\in[K]} \middle| \mathbf{W} \right) \label{eq:server_sec_step2}\\
    =&0, \label{eq:server_sec_step3}
\end{align}
where 
\begin{itemize}
    \item \eqref{eq:server_sec_step0} holds since $\mathbf{y}$ is determined by $\left\{\mathbf{x}_k\right\}_{k\in[K]}$ and $\mathbf{z}$; 
    \item \eqref{eq:server_sec_step1} is due to the independence of the channel noise $\mathbf{z}$; 
    \item \eqref{eq:server_sec_step2} follows the invariance property of mutual information (MI) under bijective transformations \cite[Theorem 3.7]{polyanskiy2025information}, i.e.,  $\mathrm{I}(X; Y)=\mathrm{I}(f(X);  g(Y))$ given that $f(\cdot)$ and $g(\cdot)$ are bijective functions; 
    \item \eqref{eq:server_sec_step3} follows from the Markov chain $\{\mathbf{W}_k\}_{k\in[K]} \circbar \mathbf{W}\circbar \left\{\mathbf{e}_k\right\}_{k\in[K]}$. Proof of this Markov chain can be found in Appendix \ref{appx: Markov_proof}. 
\end{itemize}
\vspace{-5mm}
\end{proof}

%\begin{lemma}[The permutation invariance property of MI \cite{Polyanskiy_Wu_2024}]
%    Let $X$ and $Y$ be multivariate random variables, and $f(\cdot)$ and $g(\cdot)$ be bijective functions. The mutual information remains invariant in this transformation, i.e.,  $\mathrm{I}(X; Y)=\mathrm{I}(f(X);  g(Y))$.
%\end{lemma}
%Lemma \ref{lemma:permutation invariance} is stated in Theorem 3.7 of Polyanskiy and Wu \cite{Polyanskiy_Wu_2024}. Each unique pair $(X, Y)$ is remapped to a unique pair $(f(X), g(Y))$, consequently, the one-to-one mapping preserves the dependency structure between $X$ and $Y$, the mutual information remains unchanged. To clarify and avoid confusion, it is important to note that this property does not apply to differential entropy or joint differential entropy. This is because the probability density function (PDF) must be adjusted by the Jacobian determinant to ensure that the total probability remains normalized to 1 when scaling changes occur in $f(X)$, $f(Y)$.

%%%%%%%%%%------------------------------------
%\vspace{-0.5em}
\subsection{Security at Clients}
Due to correlation among secret keys $\{\mathbf{S}_k\}_{k\in[K]}$ and wireless broadcasting, client $k$ may use its secret key $\mathbf{S}_k$ as auxiliary information to infer client $k'$ private messege $\mathbf{W}_{k'}$ from based on the broadcast signal $\mathbf{x}_{k'}$, where $k'\neq k$. 
Denote the channel coefficient from client $k'$ to $k$ by $h_{k,k'}$, and the receiver noise by $\mathbf{z}_{k,k'}\sim \mathcal{N}(\mathbf{0}, N_0 \mathbf{I}_D)$, then the received signal at client $k$ is 
$\mathbf{y}_{k,k'}=h_{k,k'} \mathbf{x}_{k'}+ \mathbf{z}_{k,k'}$.
\begin{theorem}[Perfect Privacy at Clients]
In P$^2$-AirComp, assume that channel conditions is independent from learning process, when $K\geq 3$, the clients can infer no information about the private messages beyond the aggregation, i.e., $I( \{\mathbf{W}_{k'}\}_{k'\neq k}; \{\mathbf{y}_{k,k'}\}_{k'\neq k} | \mathbf{W}_{k}, \mathbf{S}_k, \mathbf{W} )=0$.   
\end{theorem}
\begin{proof}
Given the independence of the channel conditions, when $K\geq 3$, the information leakage can be quantified by 
\begin{align}
    &I\left( \{\mathbf{W}_{k'}\}_{k'\neq k}; \{\mathbf{y}_{k,k'}\}_{k'\neq k} \middle| \mathbf{W}_{k}, \mathbf{S}_k, \mathbf{W} \right) \notag\\
    \leq&I\left( \{\mathbf{W}_{k'}\}_{k'\neq k}; \{\mathbf{x}_{k'}, \mathbf{z}_{k,k'}\}_{k'\neq k} \middle| \mathbf{W}_{k}, \mathbf{S}_k, \mathbf{W} \right) \label{eq:client_sec_step0}\\
    =&I\left( \{\mathbf{W}_{k'}\}_{k'\neq k}; \{\mathbf{e}_{k'}\}_{k'\neq k} \middle| \mathbf{W}_{k}, \mathbf{S}_k, \mathbf{W} \right)\label{eq:client_sec_step1}\\
    =&I\left( \{\mathbf{W}_{k'}\}_{k'\neq k}; \{\mathbf{e}_{k'}\}_{k'\neq k} \middle| \mathbf{W}_{k}, \mathbf{S}_k, \mathbf{W}^{-k} \right)\label{eq:client_sec_step2}\\
    =&0 \label{eq:client_sec_step3}, 
\end{align}
where $\mathbf{W}^{-k}=\sum_{k'\neq k} \mathbf{W}_{k'}$ denote the sum without client $k$'s contribution. The proof proceeds similarly to Theorem \ref{theo: server security}, except that  \eqref{eq:client_sec_step3} follows from the Markov chain $\{\mathbf{W}_{k'}\}_{k'\neq k} \circbar (\mathbf{W}_{k}, \mathbf{S}_k, \mathbf{W}^{-k}) \circbar \left\{\mathbf{e}_{k'}\right\}_{k'\neq k}$. Proof of this Markov chain can be found in Appendix \ref{appx: Markov_proof2}.
\end{proof}

\begin{remark}[Infeasibility when $K=2$]
When $K=2$, \eqref{eq:modulo_lero_sum_con} implies that $\mathbf{S}_2 = [-\mathbf{S}_1] \bmod 1$. Although privacy is preserved on the server, the secret keys no longer protect the private message across clients. In this case, while the channel noise $\mathbf{z}_{k,k'}$ can protect privacy among clients, nonzero information leakage persists.
\end{remark}

%%%%%%%%%%%%%%%------------------------------
%\vspace{-0.5em}
\section{Distortion Analysis} \label{sec: distortion-snr}
%Assume that the boundaries are large enough to ensure satisfactory privacy and that the small noise at the Rx will not make the mod fold again, then analysis of distortion, then convergence analysis.  
%The point is that the modulo wrapping affects distortion....
This section analyzes signal-dependent distortion given a noise variance. 
The noise can, first, affect the deviation of the estimation $\hat{\mathbf{W}}$ from the truth $\mathbf{W}$ and, second, induce wrapping in the modulo operation. 
The former effect is also present in conventional AirComp, whereas the latter arises specifically from the modulo operation.
In general, the impact of noise depends on both its amplitude and the transmitted signal. However, their joint effect is intricate and not transparent. 
When the signal is around $\mathbf{0}$, it can tolerate more noise before wrapping occurs than when it is near other values.
Apart from this, large noise is not necessarily detrimental. 
%Due to this, increasing the noise variance may not necessarily result in larger distortion. 
An extreme example arises when the noise is concentrated around points of the form $\boldsymbol{l} + \frac{1}{2}\cdot \mathbf{1}$, $\boldsymbol{l} \in \mathbb{Z}^D$, in which case accurate estimation can still be achieved regardless of the magnitude of $\boldsymbol{l}$.

%\vspace{-0.5em}
%%%%%%%%%%------------------------------------
\subsection{Pointwise Distortion in Expectation}
Define the effective noise in the estimation \eqref{eq: noisy sum} by $\mathbf{n}=\frac{\mathbf{z}}{\sqrt{P}}$, then $\mathbf{n}\sim\mathcal{N}(\mathbf{0}, \sigma_{\mathrm{eff}}^2\mathbf{I}_D)$ and $\sigma_{\mathrm{eff}}\triangleq \sqrt{\frac{N_0}{P}}$. When $\mathbf{W}=\mathbf{s}$, let $\mathbf{y}=\mathbf{s}+\mathbf{n}$, then $\hat{\mathbf{W}}=\hat{\mathbf{s}}= [\mathbf{y}] \mod 1$.
Assume that the entries of $\mathbf{W}$ are mutually independent.    
Let $s_i$ denote the $i$-th entry of $\mathbf{s}$. The MSE distortion between $\mathbf{W}$ and $\hat{\mathbf{W}}$ is given by 
\begin{align}
    \delta(\mathbf{s})
    &=\mathbb{E}_{\mathbf{n}}\left[ \lVert\hat{\mathbf{s}}-\mathbf{s}\rVert^2 \right] =\sum_{d=1}^{D} \mathbb{E}_{n_d}\left[ (\hat{s}_d-s_d)^2 \right] \label{eq: pre_pointwise_d0}\\
    &=\sum_{d=1}^{D} \sum_{l\in \mathbb{Z}} \int_{l-\frac{1}{2}}^{l+\frac{1}{2}} (y_d-s_d-l)^2 \phi_{\sigma_{\mathrm{eff}}} (y_d-s_d)\; dy_d  \label{eq: pre_pointwise_d1}\\
    &=\sum_{d=1}^{D} \underbrace{\sum_{l\in \mathbb{Z}} \int_{l-s_d-\frac{1}{2}}^{l-s_d+\frac{1}{2}} (n_d-l)^2 \phi_{\sigma_{\mathrm{eff}}} (n_d)\; dn_d }_{\delta(s_d)} \label{eq: pre_pointwise_d2}, 
    %&=D \sum_{l\in \mathbb{Z}} \int_{l-s_d-\frac{1}{2}}^{l-s_d+\frac{1}{2}} (n^2-2ln+l^2) \phi_{\sigma_{\mathrm{eff}}} (n_d)\; dn_d, \label{eq: pre_pointwise_d}
\end{align}
where \eqref{eq: pre_pointwise_d0} is due to the independence of noise and message across entries; in \eqref{eq: pre_pointwise_d1}, $\phi_{\sigma}(x)
\triangleq
\frac{1}{\sqrt{2\pi}\sigma}
\exp\!\left(-\frac{x^2}{2\sigma^2}\right)$; \eqref{eq: pre_pointwise_d2} follows that $n_d=y_d-s_d$; and $y_d$ and $n_d$ are the $d$-th entry in the corresponding vector.  

Let $a_{l,d}\triangleq l-s_d-\frac{1}{2}$ and $b_{l,d}\triangleq l-s_d+\frac{1}{2}$, we have
\begin{align}
&\int_{a_{l,d}}^{b_{l,d}} \phi_{\sigma_{\mathrm{eff}}} (n_d)\; dn_d \hspace{3.5mm} =\Phi\left(\frac{b_{l,d}}{\sigma_{\mathrm{eff}}}\right)-\Phi\left(\frac{a_{l,d}}{\sigma_{\mathrm{eff}}}\right), \label{eq: subterm1}\\
&\int_{a_{l,d}}^{b_{l,d}} n_d \phi_{\sigma_{\mathrm{eff}}} (n_d)\; dn_d
\hspace{-0.5mm}=\sigma_{\mathrm{eff}}\hspace{-0.5mm}\left( \hspace{-0.5mm}\psi\hspace{-0.5mm}\left(\frac{a_{l,d}}{\sigma_{\mathrm{eff}}}\right)
\hspace{-0.5mm}-\hspace{-0.5mm}\psi\hspace{-0.5mm}\left(\frac{b_{l,d}}{\sigma_{\mathrm{eff}}}\right) \hspace{-0.5mm} \right), \label{eq: subterm2}\\
&\int_{a_{l,d}}^{b_{l,d}} n_d^2 \phi_{\sigma_{\mathrm{eff}}} (n_d)\; dn_d=
\sigma_{\mathrm{eff}}\bigg( a_{l,d}\psi\left(\frac{a_{l,d}}{\sigma_{\mathrm{eff}}}\right) \notag\\
&\hspace{2mm}-b_{l,d}\psi\left(\frac{b_{l,d}}{\sigma_{\mathrm{eff}}}\right)\bigg) +\sigma_{\mathrm{eff}}^2 \left(\Phi\left(\frac{b_{l,d}}{\sigma_{\mathrm{eff}}}\right) - \Phi\left(\frac{a_{l,d}}{\sigma_{\mathrm{eff}}}\right) \right), \label{eq: subterm3}
\end{align}
where $\psi(x) = \frac{1}{\sqrt{2\pi}} e^{-\frac{1}{2}x^2}$ and $\Phi(x)= \int_{-\infty}^{x} \psi(t)\, dt$ denote the probability density function (PDF) and cummulative density function (CDF) of $\mathcal{N}(0,1)$, respectively. 
The calculation of \eqref{eq: subterm2} and \eqref{eq: subterm3} uses the fact that $\frac{d(e^{-n_d/(2\sigma^2)} )}{dn_d}=-\frac{n_d}{\sigma^2} e^{-\frac{n_d}{2\sigma^2}}$ and \eqref{eq: subterm3} further applies integration by parts formula. 

Substitute \eqref{eq: subterm1}-\eqref{eq: subterm3} into \eqref{eq: pre_pointwise_d2}, we obtain the close-form of the pointwise distortion,
\begin{align}
    &\delta(\mathbf{s})=\sum_{d=1}^{D}\sum_{l\in\mathbb{Z}} \bigg\{\hspace{-1mm} 
    \left(\sigma_{\mathrm{eff}}^2+l^2\right) 
    \left( \Phi\left(\frac{b_{l,d}}{\sigma_{\mathrm{eff}}}\right)-\Phi\left(\frac{a_{l,d}}{\sigma_{\mathrm{eff}}}\right) \right) \notag\\
    &\hspace{-0.5mm}+\hspace{-0.5mm}(a_{l,d}-2l)\sigma_{\mathrm{eff}} \psi\left(\frac{a_{l,d}}{\sigma_{\mathrm{eff}}}\right) 
    \hspace{-0.5mm}-\hspace{-0.5mm}(b_{l,d}-2l)\sigma_{\mathrm{eff}} \psi\left(\frac{b_{l,d}}{\sigma_{\mathrm{eff}}}\right) \hspace{-1mm} \bigg\}.
\end{align}

%%%%%%%--------------------------------------
%\vspace{-0.5em}
\subsection{Bounds on Distortion in Expectation}
This section presents the achievable upper and lower bounds of the pointwise MSE distortion in P$^2$-AirComp. We show that, when applying a modulo operation to a Gaussian-distributed random variable, the second-moment integral of the resulting distribution is minimized when the center of the modulo interval aligns with the mean of the Gaussian distribution, and the integral increases monotonically as the deviation between the two centers increases.

\begin{theorem}
For arbitrarily distributed $\mathbf{W}$ within hypercube $[-a,a]^D$, under zero-mean Gaussian channel noise $ \mathbf{z} \sim \mathcal{N}(\mathbf{0}, N_0 \mathbf{I}_D)$, given common power scaling factor $P$, by adopting P$^2$-AirComp, it yields that 
%the pointwise distortion in expectation resulting from  is bounded by
\begin{align}
   D\delta(0) \leq \delta(\mathbf{s})\leq D\delta(a).
\end{align}
%where 
%\begin{align*}
%    \delta(0)&=,\\
%    \delta(a)&=
%\end{align*}
Moreover, the per-dim (per-dimension) distortion satisfies $\delta(s'_d)<\delta(s_d)$ if $\lvert s'_d\rvert<\lvert s_d\rvert$, for $d=1,\cdots,D$.

\end{theorem}

\begin{proof}
Taking the derivative of $\delta(s_d)$ in \eqref{eq: pre_pointwise_d2} with respect to $s_d$ and using Leibniz rule, we get 
\begin{align}
    &\frac{d \delta(s_d)}{d s_d} = \sum_{l\in\mathbb{Z}} (b_{l,d}-l)^2 \phi_{\sigma_{\mathrm{eff}}} (b_{l,d}) \frac{d b_{l,d}}{d s_d} \notag\\
    &\hspace{3cm} - \sum_{l\in\mathbb{Z}} (a_{l,d}-l)^2 \phi_{\sigma_{\mathrm{eff}}} (a_{l,d}) \frac{d a_{l,d}}{d s_d}, 
        \label{eq: dI/ds}
\end{align}
by the definition of $a_{l,d}$ and $b_{l,d}$, $\frac{d b_{l,d}}{d s_d}=\frac{d a_{l,d}}{d s_d}=-1$. Using this fact substituting $a_{l,d}$ and $b_{l,d}$ into \eqref{eq: dI/ds}, we obtain
\begin{align}
    %&\sum_{l\in\mathbb{Z}} \left( (a_{l,d}-l)^2 \phi_{\sigma_{\mathrm{eff}}} (a_{l,d})- (b_{l,d}-l)^2 \phi_{\sigma_{\mathrm{eff}}} (b_{l,d}) \right)\\
    &\sum_{l\in\mathbb{Z}} \left( \hspace{-1mm}\left(s_d+\frac{1}{2}\right)^2 \hspace{-1mm}\phi_{\sigma_{\mathrm{eff}}} (a_{l,d})-\hspace{-0.5mm} \left(s_d-\frac{1}{2}\right)^2 \hspace{-1mm}\phi_{\sigma_{\mathrm{eff}}} (b_{l,d}) 
    \hspace{-0.5mm}\right)\\
    %=&  \left(s_d+\frac{1}{2}\right)^2 \hspace{-1mm} \sum_{l\in\mathbb{Z}}\phi_{\sigma_{\mathrm{eff}}} (a_{l,d})- \hspace{-0.5mm}\left(s_d-\frac{1}{2}\right)^2 \hspace{-1mm} \sum_{l\in\mathbb{Z}}\phi_{\sigma_{\mathrm{eff}}} (b_{l,d}) \\
    %&=  \left(s_d+\frac{1}{2}\right)^2 \hspace{-1mm} \sum_{j\in\mathbb{Z}}\phi_{\sigma_{\mathrm{eff}}} (a_j+1)- \hspace{-0.5mm}\left(s_d-\frac{1}{2}\right)^2 \hspace{-1mm} \sum_{l\in\mathbb{Z}}\phi_{\sigma_{\mathrm{eff}}} (b_{l,d}) \label{eq: change counter} \\
    =& \hspace{-0.5mm} \left(\hspace{-0.5mm}s_d+\frac{1}{2}\right)^2 \hspace{-1.5mm} \sum_{l'\in\mathbb{Z}} \hspace{-0.5mm} \phi_{\sigma_{\mathrm{eff}}} (b_{l',d})- \hspace{-0.5mm}\left(s_d-\frac{1}{2}\right)^2 \hspace{-1.5mm} \sum_{l\in\mathbb{Z}}\phi_{\sigma_{\mathrm{eff}}} (b_{l,d}) \label{eq: change aj}\\
    =&  2 s_d \sum_{l\in\mathbb{Z}}\phi_{\sigma_{\mathrm{eff}}} \left(l-s_d+\frac{1}{2}\right), 
\end{align}
where \eqref{eq: change aj} applies the change of the counter by letting $l'=l-1$ and uses $b_{l',d}=a_{l'}+1$. 
Since $\phi_{\sigma_{\mathrm{eff}}}(\cdot)>0$, we have 
\begin{align}
    \mathrm{sign}\left(\frac{d \delta(s_d)}{d s_d} \right)=\mathrm{sign}\left(s_d \right)=
    \begin{cases}
        -, \hspace{3mm} s_d<0\\
        0, \hspace{4mm} s_d=0\\
        +, \hspace{3mm} s_d>0\\
    \end{cases}.
    \label{eq: direvative_sign}
\end{align}
\iffalse
\begin{figure*}
    \centering
    \begin{minipage}[b]{0.32\textwidth}
        \centering
        \includegraphics[width=\linewidth]{dist_snr_dense.pdf}
        \caption{Theoretical MSE curves ($P_X=+\infty$).}
        \label{fig:dist-snr}
    \end{minipage}
    \hfill
    \begin{minipage}[b]{0.32\textwidth}
        \centering
        \includegraphics[width=\linewidth]{dist_privacy.pdf}
        \caption{Theoretical Privacy-MSE trade-off ($P_X=+\infty$)}
        \label{fig:dist-privacy}
    \end{minipage}
    \hfill
    \begin{minipage}[b]{0.32\textwidth}
        \centering
        \includegraphics[width=\linewidth]{dist_privacy_practical.pdf}
        \caption{Practical Privacy-MSE trade-off ($P_X=$)}
        \label{fig:test_acc_noniid_new}
    \end{minipage}
    \vspace{-3mm}
\end{figure*}
\fi
This implies that the per-dim distortion $\delta(s_d)$ monotoneously increases as $s_d$ moves from $0$ toward the boundary.
By using the change of variable formula in the integral, it can be verified that $\delta(-s_d)=\delta(s_d)$. Combined these facts, it can be concluded the $\delta(\mathbf{s})$ achieves its maximum at the boundary points $\mathbf{s} = \{a,-a\}^D$ and its minimum at $\mathbf{s} = \mathbf{0}$.  
\end{proof}

%%%%%%%%%%------------------------------------
%\vspace{-0.5em}
\section{Simulations}
This section presents numerical results for the proposed P$^2$-AirComp and compares its performance with three state-of-the-art methods under perfect CSIT.
\begin{itemize}
    \item Private AirComp with independent privacy noises with $\mathcal{N}(\mathbf{0}, \sigma_{\mathrm{indp}}^2 \mathbf{I}_D)$.
    \item Private AirComp with correlated privacy noises with $\mathcal{N}(\mathbf{0}, \sigma_{\mathrm{corr}}^2 \mathbf{I}_D)$. The generation matrix is $\mathbf{G}_{\mathrm{corr}} = \frac{\sigma_{\mathrm{corr}}}{\sqrt{5}} \left( 2\mathbf{I}_K - \mathbf{I}_K^{+1} \right)$, where $\mathbf{I}_K^{+1}$ denotes the one-step right circular shift of $\mathbf{I}_K$. 
    \item Private AirComp with zero-sum correlated privacy noises. The generation of zero-sum noises follows \cite[Construction IV.3]{weng2025codingenforcedrobustsecureaggregation} with $\gamma=1$, $\lambda^2=\sigma_{\mathrm{zero}}^2$. 
\end{itemize}
For all methods, each client transmits  $\mathbf{W}_k$, and the protocols perform as described in Section~\ref{sec: AirComp model}.
The system parameters are set as follows: $K=10$, $D=10$, and $a=\frac{1}{3}$. The channel coefficients are modeled as real-valued Rician fading to capture both line-of-sight (LoS) and scattered components. 
The channel between user $k$ and the receiver is given by $h_k = \sqrt{\frac{\kappa}{\kappa+1}} + \sqrt{\frac{1}{\kappa+1}}\iota $
where the Rician factor is set to $\kappa=5\mathrm{dB}$ and $\iota \sim \mathcal{N}(0,1)$. The additive channel noise is modeled as $\mathbf{z} \sim \mathcal{N}(\mathbf{0}, N_0 \mathbf{I}_D)$ with $N_0=1$. All users are subject to a common transmit power constraint $\frac{P_X}{N_0} \leq 15\text{dB}$, and the scaling factor $P$ is chosen accordingly to ensure feasibility.

%%%%%%%%%%-----------------------------------
\vspace{-0.5em}
\subsection{Pointwise MSE}
\begin{figure}
    \centering
    \begin{minipage}[b]{0.5\textwidth}
        \centering
        \includegraphics[width=\linewidth]{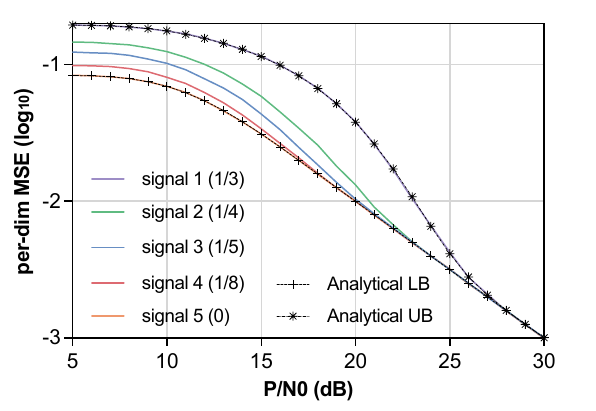}
        \vspace{-7mm}
        \caption{Pointwise MSE curves for various transmitted signals, compared with the derived analytical upper and lower bounds.}
        \vspace{-3mm}
        \label{fig:dist-snr}
    \end{minipage}
\end{figure}

In Fig.~\ref{fig:dist-snr}, we plot both the analytical bounds and the pointwise MSE curves for different transmit signals $\mathbf{s}=o\cdot\mathbf{I}_D$, where $o\in \{ 0, \frac{1}{8}, \frac{1}{5}, \frac{1}{4}, \frac{1}{3}\}$. 
It can be observed that the MSE decreases with increasing $\frac{P}{N_0}$, i.e., decreasing effective noise variance $\sigma_{\mathrm{eff}}^2=\frac{N_0}{P}$, for all transmitted signals. 
The signal closer to $\mathbf{0}$ consistently achieve lower distortion compared to those near the boundary. This observation is consistent with the theoretical insights.
Moreover, the analytical bounds tightly characterize the simulation results for all $\frac{P}{N_0}$, validating the accuracy of the theoretical analysis.
As $\frac{P}{N_0}$ increases, the performance gap between different signals gradually diminishes, and all curves converge toward the same regime, indicating that the effect of wrapping becomes negligible in the high-SNR region.

%%%%%%%%%%-----------------------------------

%%%%%%%%%%-----------------------------------

\begin{figure}
    \centering
    \begin{minipage}[b]{0.5\textwidth}
        \centering
        \includegraphics[width=\linewidth]{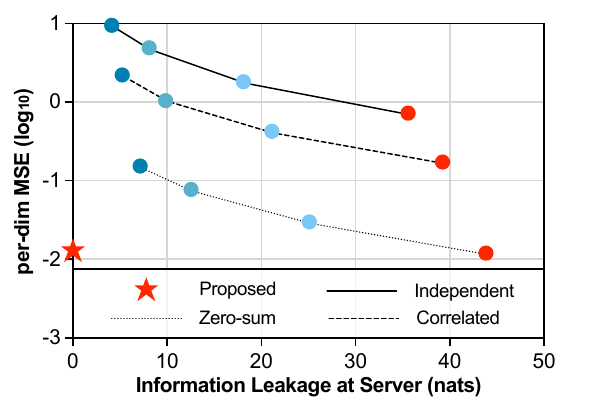}
        \vspace{-8mm}
        \caption{MSE versus mutual information leakage of the proposed P$^2$-AirComp compared with different methods under a transmit power constraint. Identical symbol colors indicate equal privacy noise power.  }
        \vspace{-5.5mm}
        \label{fig:dist-mi}
    \end{minipage}
\end{figure}

%%%%%%%%%%-----------------------------------
%\vspace{-0.5em}
\subsection{Privacy-Utility Comparison}
In Fig.~\ref{fig:dist-mi}, we plot the per-dimension MSE versus the information leakage at the server. The information leakage is quantified by 
$I( \{\mathbf{W}_k\}_{k\in[K]}; \left\{\mathbf{x}_k\right\}_{k\in[K]} \vert \mathbf{W} )$. For the zero-sum noise scheme, the transmitted signals are first projected onto a $K-1$-dimensional subspace, after which the mutual information is evaluated.
Assume $\mathbf{W}_k\sim\mathcal{N}(0,0.01\mathbf{I}_D)$, then $\mathbf{W}\sim\mathcal{N}(0,0.1\mathbf{I}_D)$. 
As shown in Fig. \ref{fig:dist-mi}, P$^2$-AirComp attains perfect privacy and per-dim MSE of $10^{-1.8839}$ under the power constraint. In contrast, none of the benchmark methods achieves zero information leakage.
Although zero-sum noise can achieve comparable MSE performance, it significantly sacrifices mutual information privacy.
Under the transmit power constraint, a clear tradeoff between MSE and privacy can be observed for all three benchmark schemes: reducing information leakage leads to increased distortion.
Although zero-sum noise is designed to cancel out, a trade-off arises between the privacy noise and the common power scaling factor. Increasing the privacy noise requires a reduction in the common power scaling factor under the power constraint, which in turn weakens the system’s ability to combat channel noise and degrades the computation accuracy.  
Similar effects arise in independent and correlated noise schemes, but in addition, the residual privacy noise increases with stronger privacy, inherently degrading estimation accuracy.
Under the same noise variance, zero-sum noise yields higher information leakage and lower MSE than correlated noise, while correlated noise results in higher information leakage and lower MSE than independent noise. This is because a stronger correlation cancels more privacy noise, thereby improving estimation accuracy, but reducing uncertainty in protecting private messages. 

These results demonstrate the effectiveness of the proposed cross-domain design, which leverages the modulo operation to preserve privacy while maintaining small distortion.

%%%%%%%%%%------------------------------------
%\vspace{-0.5em}
\section{Concluding Remarks}
This paper proposes P$^2$-AirComp, a protocol that achieves perfect privacy and accurate over-the-air computation within a finite power budget. By combining real-field and modulo operations, the proposed method can potentially overcome the privacy-distortion trade-off under a finite power constraint, a limitation that existing methods are unable to resolve. Rigorous proofs establishing perfect information-theoretic privacy with correlated secret keys are provided. Moreover, the sum-dependent distortion is derived in closed form, along with explicit upper and lower bounds. Analytical and numerical results demonstrate that P$^2$-AirComp achieves zero information leakage while maintaining low MSE, thereby outperforming existing methods.

%%%%%%%%%%%%----------------------------
\balance
\bibliographystyle{IEEEtran.bst}
\bibliography{IEEEabrv,reference}

%%%%%%%%%%------------------------------------
\begin{appendices}
\section{Proof of Markov Chain 1} \label{appx: Markov_proof}
%Consider the first $K-1$ secret keys $\{\mathbf{S}_k\}_{k=1}^{K-1}$, they are i.i.d. uniformly distributed within the interval $[0, 1)$, a
As any $K-1$ rows in $\mathbf{G}_S$ are of full rank $K-1$, the first $K-1$ generated messages $\{\mathbf{e}_k\}_{k=1}^{K-1}$ are also i.i.d. uniformly distributed within the interval $[0, 1)$, and the messages $\{\mathbf{e}_k\}_{k=1}^{K-1}$ are mutually independent from the private message $\{\mathbf{W}_k\}_{k=1}^{K}$.  
Following \eqref{eq:modulo_lero_sum_con},
%\begin{align}
%    \mathbf{S}_k=\left[-\sum_{k=1}^{K-1} \mathbf{S}_k \right]\hspace{-2mm} %\mod 1,
%\end{align}
%i.e., 
the last secret key $\mathbf{S}_k$ is fully determined by $\{\mathbf{S}_k\}_{k=1}^{K-1}$, hence,  
\begin{align}
    \mathbf{e}_K
    =\left[ \mathbf{W}-\sum_{k=1}^{K-1} \mathbf{e}_k \right]\hspace{-2mm} \mod 1.
    \label{eq: eK}
\end{align}
Then, it can be derived that 
\begin{align}
    &p\left(\left\{\mathbf{e}_k\right\}_{k\in[K]}\middle| \{\mathbf{W}_k\}_{k\in[K]}\right) \notag\\
    =&p\left(\left\{\mathbf{e}_k\right\}_{k=1}^{K-1}\middle| \{\mathbf{W}_k\}_{k\in[K]}\right) \notag\\
    &\hspace{2cm}\cdot p\left( \mathbf{e}_K \middle| \{\mathbf{W}_k\}_{k\in[K]},\left\{\mathbf{e}_k\right\}_{k=1}^{K-1}\right)\\
    =&p\left(\left\{\mathbf{e}_k\right\}_{k=1}^{K-1}\right) 
    p\left( \mathbf{e}_K \middle| \mathbf{W}, \{\mathbf{W}_k\}_{k\in[K]},\left\{\mathbf{e}_k\right\}_{k=1}^{K-1}\right)\\
    =&p\left(\left\{\mathbf{e}_k\right\}_{k=1}^{K-1}\right) 
    p\left( \mathbf{e}_K \middle| \mathbf{W},\left\{\mathbf{e}_k\right\}_{k=1}^{K-1}\right) \label{eq: markov_left_-1}\\
    =&p\left(\left\{\mathbf{e}_k\right\}_{k\in[K]} \middle| \mathbf{W}\right),
    \label{eq: markov_left}
\end{align}
where \eqref{eq: markov_left_-1} owes to \eqref{eq: eK}. Moreover, 
\begin{align}
    &p\left(\left\{\mathbf{e}_k\right\}_{k\in[K]}\middle| \{\mathbf{W}_k\}_{k\in[K]}\right) \notag\\
    =&p\left(\left\{\mathbf{e}_k\right\}_{k\in[K]}\middle| \mathbf{W}, \{\mathbf{W}_k\}_{k\in[K]}\right).
    \label{eq: markov_right}
\end{align}
Combine \eqref{eq: markov_left} and \eqref{eq: markov_right}, we get
\begin{align}
    p\left(\left\{\mathbf{e}_k\right\}_{k\in[K]} \middle| \mathbf{W}\right)
    \hspace{-1mm}
    =p\left(\left\{\mathbf{e}_k\right\}_{k\in[K]}\middle| \mathbf{W}, \{\mathbf{W}_k\}_{k\in[K]}\right),
\end{align}
which proves $\{\mathbf{W}_k\}_{k\in[K]} \circbar \mathbf{W}\circbar \left\{\mathbf{e}_k\right\}_{k\in[K]}$.

%%%%%%%%%%%%------------------------------------
\section{Proof of Markov Chain 2} \label{appx: Markov_proof2}
Following previous analysis in Appendix \ref{appx: Markov_proof}, since each client $k$ holds $\mathbf{S}_k$ and $\mathbf{W}_{k}$, any $K-2$ received messages are mutually independent. 
For $k_0\neq k$, it holds that
\begin{align}
    \mathbf{e}_{k_0}=\left[ \mathbf{W}^{-k}- \sum_{k'\neq k,k_0} \mathbf{e}_{k'} \right] \hspace{-2mm} \mod 1. 
\end{align}
Hence, 
\begin{align}
    &p\left(\left\{\mathbf{e}_{k'}\right\}_{k'\neq k}\middle| \{\mathbf{W}_{k'}\}_{k'\in[K]}, \mathbf{S}_k, \mathbf{W}^{-k}
    \right) \notag\\
    =&p\left(\left\{\mathbf{e}_{k'}\right\}_{k'\neq k,k_0}\middle| \{\mathbf{W}_{k'}\}_{k'\in[K]}, \mathbf{S}_k, \mathbf{W}^{-k}
    \right) \notag\\ 
    &\hspace{0mm}\cdot p\left(\mathbf{e}_{k_0}\middle| \{\mathbf{W}_{k'}\}_{k'\in[K]}, \mathbf{S}_k, \mathbf{W}^{-k},\left\{\mathbf{e}_{k'}\right\}_{k'\neq k,k_0}\right) \\
    =&p\left(\left\{\mathbf{e}_{k'}\right\}_{k'\neq k,k_0}\right) \notag\\
    &\hspace{10mm}\cdot p\left(\mathbf{e}_{k_0}\middle| \mathbf{W}_{k}, \mathbf{S}_k, \mathbf{W}^{-k},\left\{\mathbf{e}_{k'}\right\}_{k'\neq k,k_0}\right) \\
    =&p\left(\left\{\mathbf{e}_{k'}\right\}_{k'\neq k}\middle| \mathbf{W}_{k}, \mathbf{S}_k, \mathbf{W}^{-k}\right),
\end{align}
which proves $\{\mathbf{W}_{k'}\}_{k'\neq k} \circbar (\mathbf{W}_{k}, \mathbf{S}_k, \mathbf{W}^{-k}) \circbar \left\{\mathbf{e}_{k'}\right\}_{k'\neq k}$.

%\begin{align}
%    &p\left(\left\{\mathbf{e}_{k'}\right\}_{k'\neq k}\middle| %\mathbf{W}_{k}, \mathbf{S}_k, %\mathbf{W}^{-k}
%    \right) \notag\\
%    =&p\left( \left\{\mathbf{e}_{k'}\right\}_{k'\neq k} \middle| \%{\mathbf{W}_{k'}\}_{k'\in[K]}, \mathbf{S}_k, %\mathbf{W}^{-k}\right)
%\end{align}

\end{appendices}

\end{document}